\documentclass[english]{amsart}
\usepackage[T1]{fontenc}
\usepackage[latin9]{inputenc}
\usepackage{textcomp}
\usepackage{amsthm}
\usepackage{amsmath}
\usepackage{amssymb}
\usepackage{amsfonts}
\usepackage{esint}

\makeatletter

\theoremstyle{plain}
\newtheorem{thm}{\protect\theoremname}
  \theoremstyle{definition}
  \newtheorem{defn}[thm]{\protect\definitionname}
  \theoremstyle{remark}
  \newtheorem{rem}[thm]{\protect\remarkname}
  \theoremstyle{plain}
  \newtheorem{lem}[thm]{\protect\lemmaname}
  \theoremstyle{plain}
  \newtheorem{prop}[thm]{\protect\propositionname}

\DeclareMathOperator{\cat}{cat}

\makeatother

\usepackage{babel}
  \providecommand{\definitionname}{Definition}
  \providecommand{\lemmaname}{Lemma}
  \providecommand{\propositionname}{Proposition}
  \providecommand{\remarkname}{Remark}
\providecommand{\theoremname}{Theorem}

\begin{document}

\title[Low energy solutions for Schroedinger
Maxwell systems]{Low energy solutions for the semiclassical limit of Schroedinger
Maxwell systems}

\author{Marco Ghimenti}
\address[Marco Ghimenti] {Dipartimento di Matematica,
  Universit\`{a} di Pisa, via F. Buonarroti 1/c, 56127 Pisa, Italy}
\email{marco.ghimenti@dma.unipi.it.}
\author{Anna Maria Micheletti}
\address[Anna Maria Micheletti] {Dipartimento di Matematica,
  Universit\`{a} di Pisa, via F. Buonarroti 1/c, 56127 Pisa, Italy}
\email{a.micheletti@dma.unipi.it.}

\begin{abstract}
We show that the number of solutions of Schroedinger Maxwell system
on a smooth bounded domain $\Omega\subset\mathbb{R}^{3}$. depends
on the topological properties of the domain. In particular we consider
the Lusternik-Schnirelmann category and the Poincar\'e polynomial
of the domain.
\end{abstract}
\date{\today}

\maketitle
\begin{center}
{\bf Dedicated to our friend Bernhard}
\end{center}

\section{Introduction}

Given real numbers $q>0$, $\omega>0$ we consider the following Schroedinger
Maxwell system on a smooth bounded domain $\Omega\subset\mathbb{R}^{3}$.
\begin{equation}
\left\{ \begin{array}{cc}
-\varepsilon^{2}\Delta u+u+\omega uv=|u|^{p-2}u & \text{ in }\Omega\\
-\Delta v=qu^{2} & \text{ in }\Omega\\
u,v=0 & \text{ on }\partial\Omega
\end{array}\right.\label{eq:sms}
\end{equation}

This paper deals with the semiclassical limit of the system (\ref{eq:sms}),
i.e. it is concerned with the problem of finding solutions of (\ref{eq:sms})
when the parameter $\varepsilon$ is sufficiently small. This problem
has some relevance for the understanding of a wide class of quantum
phenomena. We are interested in the relation between the number of
solutions of (\ref{eq:sms}) and the topology of the bounded set $\Omega$.
In particular we consider the Lusternik Schnirelmann category $\cat\Omega$
of $\Omega$ in itself and its Poincar\'e polynomial $P_{t}(\Omega)$.

Our main results are the following.
\begin{thm}
\label{thm:1}Let $4<p<6$. For $\varepsilon$ small enough there
exist at least $\cat(\Omega)$ positive solutions of (\ref{eq:sms}).
\end{thm}

\begin{thm}
\label{thm:2}Let $4<p<6$. Assume that for $\varepsilon$ small enough
all the solutions of problem (\ref{eq:sms}) are non- degenerate.
Then there are at least $2P_{1}(\Omega)-1$ positive solutions.
\end{thm}
Schroedinger Maxwell systems recently received considerable attention
from the mathematical community. In the pioneering paper \cite{BF}
Benci and Fortunato studied system (\ref{eq:sms}) when $\varepsilon=1$
and without nonlinearity. Regarding the system in a semiclassical
regime Ruiz \cite{R} and D'Aprile-Wei \cite{DW1}
showed the existence of a family of radially symmetric solutions respectively
for $\Omega=\mathbb{R}^{3}$ or a ball. D'Aprile-Wei
\cite{DW2} also proved the existence of clustered solutions in the
case of a bounded domain $\Omega$ in $\mathbb{R}^{3}$. 

Recently, Siciliano \cite{S} relates the number of solution with
the topology of the set $\Omega$ when $\varepsilon=1$, and the nonlinearity
is a pure power with exponent $p$ close to the critical exponent
$6$. Moreover, in the case $\varepsilon=1$, many authors proved
results of existence and non existence of solution of (\ref{eq:sms})
in presence of a pure power nonlinearity $|u|^{p-2}u$, $2<p<6$ or
more general nonlinearities \cite{AR,ADP,AP,BJL,DM,IV,K,PS,WZ}.

In a forthcoming paper \cite{GM1}, we aim to use our approach to
give an estimate on the number of low energy solutions for Klein
Gordon Maxwell systems on a Riemannian manifold in terms of the topology
of the manifold and some information on the profile of the low energy
solutions.

In the following we always assume $4<p<6$.

\section{Notations and definitions}

In the following we use the following notations. 
\begin{itemize}
\item $B(x,r)$ is the ball in $\mathbb{R}^{3}$ centered in $x$ with radius
$r$.
\item The function $U(x)$ is the unique positive spherically symmetric
function in $\mathbb{R}^{3}$ such that
\[
-\Delta U+U=U^{p-1}\text{ in }\mathbb{R}^{3}
\]
we remark that $U$ and its first derivative decay exponentially at
infinity.
\item Given $\varepsilon>0$ we define $U_{\varepsilon}(x)=U\left(\frac{x}{\varepsilon}\right)$.
\item We denote by $\text{supp }\varphi$ the support of the function $\varphi$.
\item We define 
\[
m_{\infty}=\inf_{\int_{\mathbb{R}^{3}}|\nabla v|^{2}+v^{2}dx=
|v|_{L^{p}(\mathbb{R}^{3})}^{p}}\frac{1}{2}\int_{\mathbb{R}^{3}}|\nabla v|^{2}+v^{2}dx
-\frac{1}{p}|v|_{L^{p}(\mathbb{R}^{3})}^{p}
\]

\item We also use the following notation for the different norms for $u\in H_{g}^{1}(M)$:
\begin{eqnarray*}
\|u\|_{\varepsilon}^{2}=\frac{1}{\varepsilon^{3}}\int_{M}\varepsilon^{2}|\nabla u|^{2}+u^{2}dx &  & 
|u|_{\varepsilon,p}^{p}=\frac{1}{\varepsilon^{3}}\int_{\Omega}|u|^{p}dx\\
\|u\|_{H_{0}^{1}}^{2}=\int_{\Omega}|\nabla u|^{2}dx &  & |u|_{p}^{p}=\int_{\Omega}|u|^{p}dx
\end{eqnarray*}
and we denote by $H_{\varepsilon}$ the Hilbert space $H_{0}^{1}(\Omega)$
endowed with the $\|\cdot\|_{\varepsilon}$ norm.\end{itemize}
\begin{defn}
Let $X$ a topological space and consider a closed subset $A\subset X$.
We say that $A$ has category $k$ relative to $X$ ($\cat_{M}A=k$)
if $A$ is covered by $k$ closed sets $A_{j}$, $j=1,\dots,k$, which
are contractible in $X$, and $k$ is the minimum integer with this
property. We simply denote $\cat X=\cat_{X}X$.\end{defn}
\begin{rem}
Let $X_{1}$ and $X_{2}$ be topological spaces. If $g_{1}:X_{1}\rightarrow X_{2}$
and $g_{2}:X_{2}\rightarrow X_{1}$ are continuous operators such
that $g_{2}\circ g_{1}$ is homotopic to the identity on $X_{1}$,
then $\cat X_{1}\leq\cat X_{2}$ . \end{rem}
\begin{defn}
Let X be any topological space and let $H_{k}(X)$ denotes its $k$-th
homology group with coefficients in $\mathbb{Q}$. The Poincar\'e polynomial
$P_{t}(X)$ of $X$ is defined as the following power series in $t$
\[
P_{t}(X):=\sum_{k\ge0}\left(\text{dim}H_{k}(X)\right)t^{k}
\]

Actually, if $X$ is a compact space, we have that $\text{dim}H_{k}(X)<\infty$
and this series is finite; in this case, $P_{t}(X)$ is a polynomial
and not a formal series. \end{defn}
\begin{rem}
\label{rem:morse}Let $X$ and $Y$ be topological spaces. If $f:X\rightarrow Y$
and $g:Y\rightarrow X$ are continuous operators such that $g\circ f$
is homotopic to the identity on $X$, then $P_{t}(Y)=P_{t}(X)+Z(t)$
where $Z(t)$ is a polynomial with non-negative coefficients.

These topological tools are classical and can be found, e.g., in \cite{P}
and in \cite{B}.
\end{rem}

\section{Preliminary results}

Using an idea in a paper of Benci and Fortunato \cite{BF} we define
the map $\psi:H_{0}^{1}(\Omega)\rightarrow H_{0}^{1}(\Omega)$ defined
by the equation
\begin{equation}
-\Delta\psi(u)=qu^{2}\text{ in }\Omega\label{eq:psi}
\end{equation}

\begin{lem}
\label{lem:psi}The map $\psi:H_{0}^{1}(\Omega)\rightarrow H_{0}^{1}(\Omega)$
is of class $C^{2}$ with derivatives
\begin{eqnarray}
\psi'(u)[\varphi] & = & i^{*}(2qu\varphi)\label{eq:derprima}\\
\psi''(u)[\varphi_{1},\varphi_{2}] & = & i^{*}(2q\varphi_{1}\varphi_{2})\label{eq:derseconda}
\end{eqnarray}
where the operator $i_{\varepsilon}^{*}:L^{p'},|\cdot|_{\varepsilon,p'}\rightarrow H_{\varepsilon}$
is the adjoint operator of the immersion operator $i_{\varepsilon}:H_{\varepsilon}\rightarrow L^{p},|\cdot|_{\varepsilon,p}$.\end{lem}
\begin{proof}
The proof is standard.\end{proof}
\begin{lem}
\label{lem:Tder}The map $T:H_{0}^{1}(\Omega)\rightarrow\mathbb{R}$
given by
\[
T(u)=\int_{\Omega}u^{2}\psi(u)dx
\]
is a $C^{2}$ map and its first derivative is 
\[
T'(u)[\varphi]=4\int_{\Omega}\varphi u\psi(u)dx.
\]
\end{lem}
\begin{proof}
The regularity is standard. The first derivative is 
\[
T'(u)[\varphi]=2\int u\varphi\psi(u)+\int u^{2}\psi'(u)[\varphi].
\]

By (\ref{eq:derprima}) and (\ref{eq:psi}) we have 
\begin{eqnarray*}
2q\int u\varphi\psi(u) & = & -\int\Delta(\psi'(u)[\varphi])\psi(u)=-\int\psi'(u)[\varphi]\Delta\psi(u)=\\
 & = & \int\psi'(u)[\varphi]qu^{2}
\end{eqnarray*}
and the claim follows.
\end{proof}
At this point we consider the following functional $I_{\varepsilon}\in C^{2}(H_{0}^{1}(\Omega),\mathbb{R})$.
\begin{equation}
I_{\varepsilon}(u)=\frac{1}{2}\|u\|_{\varepsilon}^{2}+\frac{\omega}{4}G_{\varepsilon}(u)-\frac{1}{p}|u^{+}|_{\varepsilon,p}^{p}\label{eq:ieps}
\end{equation}
where 
\[
G_{\varepsilon}(u)=\frac{1}{\varepsilon^{3}}\int_{\Omega}u^{2}\psi(u)dx=\frac{1}{\varepsilon^{3}}T(u).
\]

By Lemma \ref{lem:Tder} we have 
\[
I_{\varepsilon}'(u)[\varphi]=\frac{1}{\varepsilon^{3}}\int_{\Omega}\varepsilon^{2}\nabla u\nabla\varphi+u\varphi+\omega u\psi(u)\varphi-(u^{+})^{p-1}\varphi
\]
\[
I_{\varepsilon}'(u)[u]=\|u\|_{\varepsilon}^{2}+\omega G_{\varepsilon}(u)-|u^{+}|_{\varepsilon,p}^{p}
\]
 then if $u$ is a critical points of the functional $I_{\varepsilon}$
the pair of positive functions $(u,\psi(u))$ is a solution of (\ref{eq:sms}).

\section{Nehari Manifold}

We define the following Nehari set
\[
{\mathcal N}_{\varepsilon}=\left\{ u\in H_{0}^{1}(\Omega)\smallsetminus0\ :\ N_{\varepsilon}(u):=I'_{\varepsilon}(u)[u]=0\right\} 
\]
In this section we give an explicit proof of the main properties of
the Nehari manifold, although standard, for the sake of completeness
\begin{lem}
${\mathcal N}_{\varepsilon}$ is a $C^{2}$ manifold and $\inf_{{\mathcal N}_{\varepsilon}}\|u\|_{\varepsilon}>0$.\end{lem}
\begin{proof}
If $u\in{\mathcal N}_{\varepsilon}$, using that $N_{\varepsilon}(u)=0$,
and $p>4$ we have 
\[
N'_{\varepsilon}(u)[u]=2\|u\|_{\varepsilon}^{2}+4\omega G_{\varepsilon}(u)-p|u^{+}|_{\varepsilon,p}
=(2-p)\|u\|_{\varepsilon}+(4-p)\omega G_{\varepsilon}(u)<0
\]
so ${\mathcal N}_{\varepsilon}$ is a $C^{2}$ manifold.

We prove the second claim by contradiction. Take a sequence $\left\{ u_{n}\right\} _{n}\in{\mathcal N}_{\varepsilon}$
with $\|u_{n}\|_{\varepsilon}\rightarrow0$ while $n\rightarrow+\infty$.
Thus, using that $N_{\varepsilon}(u)=0$,
\[
\|u_{n}\|_{\varepsilon}^{2}+\omega G_{\varepsilon}(u_{n})=|u_{n}^{+}|_{p,\varepsilon}^{p}\le C\|u_{n}\|_{\varepsilon}^{p},
\]
so 
\[
1<1+\frac{\omega G_{\varepsilon}(u)}{\|u_{n}\|_{\varepsilon}}\le C\|u_{n}\|_{\varepsilon}^{p-2}\rightarrow0
\]
and this is a contradiction.\end{proof}
\begin{rem}
\label{rem:nehari}If $u\in{\mathcal N}_{\varepsilon}$, then 
\begin{eqnarray*}
I_{\varepsilon}(u) & = & \left(\frac{1}{2}-\frac{1}{p}\right)\|u\|_{\varepsilon}^{2}+\omega\left(\frac{1}{4}-\frac{1}{p}\right)G_{\varepsilon}(u)\\
 & = & \left(\frac{1}{2}-\frac{1}{p}\right)|u^{+}|_{p,\varepsilon}^{p}-\frac{\omega}{4}G_{\varepsilon}(u)
\end{eqnarray*}
\end{rem}
\begin{lem}
It holds Palais-Smale condition for the functional $I_{\varepsilon}$
on ${\mathcal N}_{\varepsilon}$.\end{lem}
\begin{proof}
We start proving PS condition for $I_{\varepsilon}$. Let $\left\{ u_{n}\right\} _{n}\in H_{0}^{1}(\Omega)$
such that 
\begin{eqnarray*}
I_{\varepsilon}(u_{n})\rightarrow c &  & 
\left|I'_{\varepsilon}(u_{n})[\varphi]\right|\le\sigma_{n}\|\varphi\|_{\varepsilon}\text{ where }\sigma_{n}\rightarrow0
\end{eqnarray*}
We prove that $\|u_{n}\|_{\varepsilon}$ is bounded. Suppose $\|u_{n}\|_{\varepsilon}\rightarrow\infty$.
Then, by PS hypothesis
\[
\frac{pI_{\varepsilon}(u_{n})-I'_{\varepsilon}(u_{n})[u_{n}]}{\|u_{n}\|_{\varepsilon}}
=\left(\frac{p}{2}-1\right)\|u_{n}\|_{\varepsilon}+\left(\frac{p}{4}-1\right)\frac{G_{\varepsilon}(u_{n})}{\|u_{n}\|_{\varepsilon}}\rightarrow0
\]
and this is a contradiction because $p>4$.

At this point, up to subsequence $u_{n}\rightarrow u$ weakly in $H_{0}^{1}(\Omega)$
and strongly in $L^{t}(\Omega)$ for each $2\le t<6$. Since $u_{n}$
is a PS sequence
\[
u_{n}+\omega i_{\varepsilon}^{*}(\psi(u_{n})u_{n})-i_{\varepsilon}^{*}\left((u_{n}^{+})^{p-1}\right)\rightarrow0\text{ in }H_{0}^{1}(\Omega)
\]
we have only to prove that $i_{\varepsilon}^{*}(\psi(u_{n})u_{n})\rightarrow i_{\varepsilon}^{*}(\psi(u)u)$
in $H_{0}^{1}(\Omega)$, then we have to prove that 
\[
\psi(u_{n})u_{n}\rightarrow\psi(u)u\text{ in }L^{t'}
\]
We have $|\psi(u_{n})u_{n}-\psi(u)u|_{\varepsilon,t'}\le\left|\psi(u)(u_{n}-u)\right|_{\varepsilon,t'}
+\left|\left(\psi(u_{n})-\psi(u)\right)u_{n}\right|_{\varepsilon,t'}$.
We get 
\[
\int_{\Omega}|\psi(u_{n})-\psi(u)|^{\frac{t}{t-1}}|u_{n}|^{\frac{t}{t-1}}
\le\left(\int_{\Omega}|\psi(u_{n})-\psi(u)|^{t}\right)^{\frac{1}{t-1}}\left(\int_{\Omega}|u_{n}|^{\frac{t}{t-2}}\right)^{\frac{t-2}{t-1}}\rightarrow0,
\]
thus we can conclude easily.

Now we prove PS condition for the constrained functional. Let $\left\{ u_{n}\right\} _{n}\in{\mathcal N}_{\varepsilon}$
such that
\[
\begin{array}{cc}
I_{\varepsilon}(u_{n})\rightarrow c\\
\left|I'_{\varepsilon}(u_{n})[\varphi]-\lambda_{n}N'(u_{n})[\varphi]\right|\le\sigma_{n}\|\varphi\|_{\varepsilon} & \text{ with }\sigma_{n}\rightarrow0
\end{array}
\]
In particular 
$I'_{\varepsilon}(u_{n})\left[\frac{u_{n}}{\|u_{n}\|_{\varepsilon}}\right]
-\lambda_{n}N'(u_{n})\left[\frac{u_{n}}{\|u_{n}\|_{\varepsilon}}\right]\rightarrow0$.
Then 
\[
\lambda_{n}\left\{ \left(p-2\right)\|u_{n}\|_{\varepsilon}+\left(p-4\right)\omega\frac{G_{\varepsilon}(u_{n})}{\|u_{n}\|_{\varepsilon}}\right\} \rightarrow0
\]
thus $\lambda_{n}\rightarrow0$ because $p>4$. Since 
$N'(u_{n})=u_{n}-i_{\varepsilon}^{*}(4\omega\psi(u_{n})u_{n})-pi_{\varepsilon}^{*}(|u_{n}^{+}|^{p-1})$
is bounded we obtain that $\left\{ u_{n}\right\} _{n}$ is a PS sequence
for the free functional $I_{\varepsilon}$, and we get the claim
\end{proof}
\begin{lem}
\label{lem:teps}For all $w\in H_{0}^{1}(\Omega)$ such that $|w^{+}|_{\varepsilon,p}=1$
there exists a unique positive number $t_{\varepsilon}=t_{\varepsilon}(w)$
such that $t_{\varepsilon}(w)w\in{\mathcal N}_{\varepsilon}$. \end{lem}
\begin{proof}
We define, for $t>0$ 
\[
H(t)=I_{\varepsilon}(tw)=\frac{1}{2}t^{2}\|w\|_{\varepsilon}^{2}+\frac{t^{4}}{4}\omega G_{\varepsilon}(w)-\frac{t^{p}}{p}.
\]
Thus
\begin{eqnarray}
H'(t) & = & t\left(\|w\|_{\varepsilon}^{2}+t^{2}\omega G_{\varepsilon}(w)-t^{p-2}\right)\label{eq:Hprimo}\\
H''(t) & = & \|w\|_{\varepsilon}^{2}+3t^{2}\omega G_{\varepsilon}(w)-(p-1)t^{p-2}\label{eq:Hsec}
\end{eqnarray}
By (\ref{eq:Hprimo}) there exists $t_{\varepsilon}>0$ such that
$H'(t_{\varepsilon})$. Moreover, by (\ref{eq:Hprimo}), (\ref{eq:Hsec})
and because $p>4$ we that $H''(t_{\varepsilon})<0$, so $t_{\varepsilon}$
is unique.
\end{proof}

\section{Main ingredient of the proof}

We sketch the proof of Theorem \ref{thm:1}. First of all, since the
functional $I_{\varepsilon}\in C^{2}$ is bounded below and satisfies
PS condition on the complete $C^{2}$ manifold ${\mathcal N}_{\varepsilon}$,
we have, by well known results, that $I_{\varepsilon}$ has at least
$\cat I_{\varepsilon}^{d}$ critical points in the sublevel
\[
I_{\varepsilon}^{d}=\left\{ u\in H^{1}\ :\ I_{\varepsilon}(u)\le d\right\} .
\]
We prove that, for $\varepsilon$ and $\delta$ small enough, it holds
\[
\cat\Omega\le\cat\left({\mathcal N}_{\varepsilon}\cap I_{\varepsilon}^{m_{\infty}+\delta}\right)
\]
where 
\[
m_{\infty}:=\inf_{{\mathcal N}_{\infty}}\frac{1}{2}\int_{\mathbb{R}^{3}}|\nabla v|^{2}+v^{2}dx-\frac{1}{p}\int_{\mathbb{R}^{3}}|v|^{p}dx
\]
\[
{\mathcal N}_{\infty}=\left\{ v\in H^{1}(\mathbb{R}^{3})\smallsetminus\left\{ 0\right\} \ :\ \int_{\mathbb{R}^{3}}|\nabla v|^{2}
+v^{2}dx=\int_{\mathbb{R}^{3}}|v|^{p}dx\right\} .
\]
To get the inequality $\cat\Omega\le\cat\left({\mathcal N}_{\varepsilon}\cap I_{\varepsilon}^{m_{\infty}+\delta}\right)$
we build two continuous operators
\begin{eqnarray*}
\Phi_{\varepsilon} & : & \Omega^{-}\rightarrow{\mathcal N}_{\varepsilon}\cap I_{\varepsilon}^{m_{\infty}+\delta}\\
\beta & : & {\mathcal N}_{\varepsilon}\cap I_{\varepsilon}^{m_{\infty}+\delta}\rightarrow\Omega^{+}.
\end{eqnarray*}
where 
\[
\Omega^{-}=\left\{ x\in\Omega\ :\ d(x,\partial\Omega)<r\right\} 
\]
\[
\Omega^{+}=\left\{ x\in\mathbb{R}^{3}\ :\ d(x,\partial\Omega)<r\right\} 
\]
with $r$ small enough so that $\cat(\Omega^{-})=\cat(\Omega^{+})=\cat(\Omega)$.

Following an idea in \cite{BC1}, we build these operators $\Phi_{\varepsilon}$
and $\beta$ such that $\beta\circ\Phi_{\varepsilon}:\Omega^{-}\rightarrow\Omega^{+}$
is homotopic to the immersion $i:\Omega^{-}\rightarrow\Omega^{+}$.
By the properties of Lusternik Schinerlmann category we have
\[
\cat\Omega\le\cat\left({\mathcal N}_{\varepsilon}\cap I_{\varepsilon}^{m_{\infty}+\delta}\right)
\]
which ends the proof of Theorem \ref{thm:1}. 

Concerning Theorem \ref{thm:2}, we can re-state classical results
contained in \cite{B,BC2} in the following form.
\begin{thm}
Let $I_{\varepsilon}$ be the functional (\ref{eq:ieps}) on $H^{1}(\Omega)$
and let $K_{\varepsilon}$ be the set of its critical points. If all
its critical points are non-degenerate then
\begin{equation}
\sum_{u\in K_{\varepsilon}}t^{\mu(u)}=tP_{t}(\Omega)+t^{2}(P_{t}(\Omega)-1)+t(1+t)Q(t)\label{eq:morse1}
\end{equation}
 where Q(t) is a polynomial with non-negative integer coefficients
and $\mu(u)$ is the Morse index of the critical point $u$.
\end{thm}
By Remark \ref{rem:morse} and by means of the maps $\Phi_{\varepsilon}$
and $\beta$ we have that 
\begin{equation}
P_{t}({\mathcal N}_{\varepsilon}\cap I_{\varepsilon}^{m_{\infty}+\delta})=P_{t}(\Omega)+Z(t)\label{eq:morse2}
\end{equation}
 where $Z(t)$ is a polynomial with non-negative coefficients. Provided
that $\inf_{\varepsilon}m_{\varepsilon}=:\alpha>0$, because ${\displaystyle \lim_{\varepsilon\rightarrow0}m_{\varepsilon}=m_{\infty}}$
(see \ref{eq:mepsminfty}) , we have the following relations \cite{B,BC2}
\begin{equation}
P_{t}(I_{\varepsilon}^{m_{\infty}+\delta},I_{\varepsilon}^{\alpha/2})
=tP_{t}({\mathcal N}_{\varepsilon}\cap I_{\varepsilon}^{m_{\infty}+\delta})\label{eq:morse3}
\end{equation}
\begin{equation}
P_{t}(H_{0}^{1}(\Omega),I_{\varepsilon}^{m_{\infty}+\delta}))
=t(P_{t}(I_{\varepsilon}^{m_{\infty}+\delta},I_{\varepsilon}^{\alpha/2})-t)\label{eq:morse4}
\end{equation}
\begin{equation}
\sum_{u\in K_{\varepsilon}}t^{\mu(u)}
=P_{t}(H_{0}^{1}(\Omega),I_{\varepsilon}^{m_{\infty}+\delta}))
+P_{t}(I_{\varepsilon}^{m_{\infty}+\delta},I_{\varepsilon}^{\alpha/2})+(1+t)\tilde{Q}(t)\label{eq:morse5}
\end{equation}
 where $\tilde{Q}(t)$ is a polynomial with non-negative integer coefficients.
Hence, by (\ref{eq:morse2}), (\ref{eq:morse3}), (\ref{eq:morse4}),
(\ref{eq:morse5}) we obtain (\ref{eq:morse1}). At this point, evaluating
equation (\ref{eq:morse1}) for $t=1$ we obtain the claim of Theorem
\ref{thm:2}

\section{The map $\Phi_{\varepsilon}$ }

For every $\xi\in\Omega^{-}$ we define the function
\[
W_{\xi,\varepsilon}(x)=U_{\varepsilon}(x-\xi)\chi(|x-\xi|)
\]
where $\chi:\mathbb{R}^{+}\rightarrow\mathbb{R}^{+}$ where $\chi\equiv1$
for $t\in[0,r/2)$, $\chi\equiv0$ for $t>r$ and $|\chi'(t)|\le2/r$.

We can define a map
\begin{eqnarray*}
\Phi_{\varepsilon} & : & \Omega^{-}\rightarrow{\mathcal N}_{\varepsilon}\\
\Phi_{\varepsilon}(\xi) & = & t_{\varepsilon}(W_{\xi,\varepsilon})W_{\xi,\varepsilon}
\end{eqnarray*}

\begin{rem}
\label{w}We have that the following limits hold uniformly with respect
to $\xi\in\Omega$
\begin{eqnarray*}
\|W_{\varepsilon,\xi}\|_{\varepsilon} & \rightarrow & \|U\|_{H^{1}(\mathbb{R}^{3})}\\
|W_{\varepsilon,\xi}|_{\varepsilon,t} & \rightarrow & \|U\|_{L^{t}(\mathbb{R}^{3})}\text{ for all }2\le t\le6
\end{eqnarray*}
\end{rem}
\begin{lem}
\label{lem:stimaGeps}There exists $\bar{\varepsilon}>0$ and a constant
$c>0$ such that 
\[
G_{\varepsilon}(W_{\varepsilon,\xi})=\frac{1}{\varepsilon^{3}}\int_{\Omega}qW_{\varepsilon,\xi}^{2}(x)\psi(W_{\varepsilon,\xi})dx<c\varepsilon^{2}
\]
\end{lem}
\begin{proof}
It holds 
\begin{eqnarray*}
\|\psi(W_{\varepsilon,\xi})\|_{H_{0}^{1}(\Omega)}^{2} & = 
& \int_{\Omega}qW_{\varepsilon,\xi}^{2}(x)\psi(W_{\varepsilon,\xi})dx
\le q\|\psi(W_{\varepsilon,\xi})\|_{L^{6}(\Omega)}\left(\int_{\Omega}W_{\varepsilon,\xi}^{12/5}dx\right)^{5/6}\\
 & \le & c\|\psi(W_{\varepsilon,\xi})\|_{H_{0}^{1}(\Omega)}
 \left(\frac{1}{\varepsilon^{3}}\int_{\Omega}W_{\varepsilon,\xi}^{12/5}dx\right)^{5/6}\varepsilon^{5/2}
\end{eqnarray*}
By Remark \ref{w} we have that $\|\psi(W_{\varepsilon,\xi})\|_{H_{0}^{1}(\Omega)}\le\varepsilon^{5/2}$
and the claim follows by applying again Cauchy Schwartz inequality.\end{proof}
\begin{prop}
\label{prop:phieps}For all $\varepsilon>0$ the map $\Phi_{\varepsilon}$
is continuous. Moreover for any $\delta>0$ there exists $\varepsilon_{0}=\varepsilon_{0}(\delta)$
such that, if $\varepsilon<\varepsilon_{0}$ then $I_{\varepsilon}\left(\Phi_{\varepsilon}(\xi)\right)<m_{\infty}+\delta$.\end{prop}
\begin{proof}
It is easy to see that $\Phi_{\varepsilon}$ is continuous because
$t_{\varepsilon}(w)$ depends continously on $w\in H_{0}^{1}$.

At this point we prove that $t_{\varepsilon}(W_{\varepsilon,\xi})\rightarrow1$
uniformly with respect to $\xi\in\Omega$. In fact, by Lemma \ref{lem:teps}
$t_{\varepsilon}(W_{\varepsilon,\xi})$ is the unique solution of
\[
\|W_{\varepsilon,\xi}\|_{\varepsilon}^{2}+t^{2}\omega G_{\varepsilon}(W_{\varepsilon,\xi})-t^{p-2}|W_{\varepsilon,\xi}|_{\varepsilon,p}^{p}=0.
\]
By Remark \ref{w} and Lemma \ref{lem:stimaGeps} we have the claim.

Now, we have 
\[
I_{\varepsilon}\left(t_{\varepsilon}(W_{\varepsilon,\xi})W_{\varepsilon,\xi}\right)
=\left(\frac{1}{2}-\frac{1}{p}\right)\|W_{\varepsilon,\xi}\|_{\varepsilon}^{2}t_{\varepsilon}^{2}
+\omega\left(\frac{1}{4}-\frac{1}{p}\right)t_{\varepsilon}^{4}G_{\varepsilon}(W_{\varepsilon,\xi})
\]
Again, by Remark \ref{w} and Lemma \ref{lem:stimaGeps} we have

\[
I_{\varepsilon}\left(t_{\varepsilon}(W_{\varepsilon,\xi})W_{\varepsilon,\xi}\right)
\rightarrow\left(\frac{1}{2}-\frac{1}{p}\right)\|U\|_{H^{1}(\mathbb{R}^{3})}^{2}=m_{\infty}
\]
that concludes the proof.
\end{proof}
\begin{rem}
\label{rem:limsup}We set
\[
m_{\varepsilon}=\inf_{{\mathcal N}_{\varepsilon}}I_{\varepsilon.}
\]
By Proposition \ref{prop:phieps} we have that 

\begin{equation}
\limsup_{\varepsilon\rightarrow0}m_{\varepsilon}\le m_{\infty.}\label{eq:limsup}
\end{equation}
\end{rem}

\section{The map $\beta$}

For any $u\in{\mathcal N}_{\varepsilon}$ we can define a point $\beta(u)\in\mathbb{R}^{3}$
by 
\[
\beta(u)=\frac{\int_{\Omega}x|u^{+}|^{p}dx}{\int_{\Omega}|u^{+}|^{p}dx}.
\]
The function $\beta$ is well defined in ${\mathcal N}_{\varepsilon}$
because, if $u\in{\mathcal N}_{\varepsilon}$, then $u^{+}\neq0$.

We have to prove that, if $u\in{\mathcal N}_{\varepsilon}\cap I_{\varepsilon}^{m_{\infty}+\delta}$
then $\beta(u)\in\Omega^{+}$.

Let us consider partitions of $\Omega$. For a given $\varepsilon>0$
we say that a finite partition ${\mathcal P}_{\varepsilon}=\left\{ P_{j}^{\varepsilon}\right\} _{j\in\Lambda_{\varepsilon}}$
of $\Omega$ is a ``good'' partition
if: for any $j\in\Lambda_{\varepsilon}$ the set $P_{j}^{\varepsilon}$
is closed; $P_{i}^{\varepsilon}\cap P_{j}^{\varepsilon}\subset\partial P_{i}^{\varepsilon}\cap\partial P_{j}^{\varepsilon}$
for any $i\ne j$; there exist $r_{1}(\varepsilon),r_{2}(\varepsilon)>0$
such that there are points $q_{j}^{\varepsilon}\in P_{j}^{\varepsilon}$
for which  
$B(q_{j}^{\varepsilon},\varepsilon)\subset P_{j}^{\varepsilon}\subset B(q_{j}^{\varepsilon},r_{2}(\varepsilon))
\subset B_{g}(q_{j}^{\varepsilon},r_{1}(\varepsilon))$,
with $r_{1}(\varepsilon)\ge r_{2}(\varepsilon)\ge C\varepsilon$ for
some positive constant $C$; lastly, there exists a finite number
$\nu\in\mathbb{N}$ such that every $x\in\Omega$ is contained in
at most $\nu$ balls $B(q_{j}^{\varepsilon},r_{1}(\varepsilon))$,
where $\nu$ does not depends on $\varepsilon$.
\begin{lem}
\label{lem:gamma}There exists a constant $\gamma>0$ such that, for
any $\delta>0$ and for any $\varepsilon<\varepsilon_{0}(\delta)$
as in Proposition \ref{prop:phieps}, given any ``good'' partition
${\mathcal P}_{\varepsilon}=\left\{ P_{j}^{\varepsilon}\right\} _{j}$
of the domain $\Omega$ and for any function $u\in{\mathcal N}_{\varepsilon}\cap I_{\varepsilon}^{m_{\infty}+\delta}$
there exists, for an index $\bar{j}$ a set $P_{\bar{j}}^{\varepsilon}$
such that 
\[
\frac{1}{\varepsilon^{3}}\int_{P_{\bar{j}}^{\varepsilon}}|u^{+}|^{p}dx\ge\gamma.
\]
\end{lem}
\begin{proof}
Taking in account that $I'(u)[u]=0$ we have
\begin{eqnarray*}
\|u\|_{\varepsilon}^{2} & = & |u^{+}|_{\varepsilon,p}^{p}
-\frac{1}{\varepsilon^{3}}\int_{\Omega}\omega u^{2}\psi(u)\le|u^{+}|_{\varepsilon,p}^{p}
=\sum_{j}\frac{1}{\varepsilon^{3}}\int_{P_{j}}|u^{+}|^{p}\\
 & = & \sum_{j}|u_{j}^{+}|_{\varepsilon,p}^{p}
 =\sum_{j}|u_{j}^{+}|_{\varepsilon,p}^{p-2}|u_{j}^{+}|_{\varepsilon,p}^{2}
 \le\max_{j}\left\{ |u_{j}^{+}|_{\varepsilon,p}^{p-2}\right\} \sum_{j}|u_{j}^{+}|_{\varepsilon,p}^{2}
\end{eqnarray*}
where $u_{j}^{+}$ is the restriction of the function $u^{+}$ on the
set $P_{j}$. 

At this point, arguing as in \cite[Lemma 5.3]{BBM}, we prove that there
exists a constant $C>0$ such that 
\[
\sum_{j}|u_{j}^{+}|_{\varepsilon,p}^{2}\le C\nu\|u^{+}\|_{\varepsilon}^{2},
\]
thus 
\[
\max_{j}\left\{ |u_{j}^{+}|_{\varepsilon,p}^{p-2}\right\} \ge\frac{1}{C\nu}
\]
 that conludes the proof.\end{proof}
\begin{prop}
\label{prop:conc}For any $\eta\in(0,1)$ there exists $\delta_{0}<m_{\infty}$
such that for any $\delta\in(0,\delta_{0})$ and any $\varepsilon\in(0,\varepsilon_{0}(\delta))$
as in Proposition \ref{prop:phieps}, for any function $u\in{\mathcal N}_{\varepsilon}\cap I_{\varepsilon}^{m_{\infty}+\delta}$
we can find a point $q=q(u)\in\Omega$ such that 
\[
\frac{1}{\varepsilon^{3}}\int_{B(q,r/2)}(u^{+})^{p}>\left(1-\eta\right)\frac{2p}{p-2}m_{\infty}.
\]
\end{prop}
\begin{proof}
First, we prove the proposition for $u\in{\mathcal N}_{\varepsilon}\cap I_{\varepsilon}^{m_{\varepsilon}+2\delta}$. 

By contradiction, we assume that there exists $\eta\in(0,1)$ such
that we can find two sequences of vanishing real number $\delta_{k}$
and $\varepsilon_{k}$ and a sequence of functions $\left\{ u_{k}\right\} _{k}$
such that $u_{k}\in{\mathcal N}_{\varepsilon_{k}}$, 
\begin{equation}
m_{\varepsilon_{k}}\le I_{\varepsilon_{k}}(u_{k})=\left(\frac{1}{2}-\frac{1}{p}\right)\|u_{k}\|_{\varepsilon_{k}}^{2}
+\omega\left(\frac{1}{4}-\frac{1}{p}\right)G_{\varepsilon_{k}}(u_{k})\le m_{\varepsilon_{k}}
+2\delta_{k}\le m_{\infty}+3\delta_{k}\label{eq:mepsk}
\end{equation}
 for $k$ large enough (see Remark \ref{rem:limsup}), and, for any
$q\in\Omega$, 
\[
\frac{1}{\varepsilon_{k}^{3}}\int_{B(q,r/2)}(u_{k}^{+})^{p}\le\left(1-\eta\right)\frac{2p}{p-2}m_{\infty}.
\]
By Ekeland principle and by definition of ${\mathcal N}_{\varepsilon_{k}}$
we can assume 
\begin{equation}
\left|I'_{\varepsilon_{k}}(u_{k})[\varphi]\right|\le\sigma_{k}\|\varphi\|_{\varepsilon_{k}}\text{ where }\sigma_{k}\rightarrow0.\label{eq:ps}
\end{equation}

By Lemma \ref{lem:gamma} there exists a set $P_{k}^{\varepsilon_{k}}\in{\mathcal P}_{\varepsilon_{k}}$
such that 
\[
\frac{1}{\varepsilon_{k}^{3}}\int_{P_{k}^{\varepsilon_{k}}}|u_{k}^{+}|^{p}dx\ge\gamma.
\]
 We choose a point $q_{k}\in P_{k}^{\varepsilon_{k}}$ and we define,
for $z\in\Omega_{\varepsilon_{k}}:=\frac{1}{\varepsilon_{k}}\left(\Omega-q_{k}\right)$
\[
w_{k}(z)=u_{k}(\varepsilon_{k}z+q_{k})=u_{k}(x).
\]

We have that $w_{k}\in H_{0}^{1}(\Omega_{\varepsilon_{k}})\subset H^{1}(\mathbb{R}^{3})$.
By equation (\ref{eq:mepsk}) we have 
\[
\|w_{k}\|_{H^{1}(\mathbb{R}^{3})}^{2}=\|u_{k}\|_{\varepsilon_{k}}^{2}\le C.
\]
So $w_{k}\rightarrow w$ weakly in $H^{1}(\mathbb{R}^{3})$ and strongly
in $L_{\text{loc}}^{t}(\mathbb{R}^{3})$. 

We set $\psi(u_{k})(x):=\psi_{k}(x)=\psi_{k}(\varepsilon_{k}z+q_{k}):=\tilde{\psi}_{k}(z)$
where $x\in\Omega$ and $z\in\Omega_{\varepsilon_{k}}$. It is easy
to verify that 
\[
-\Delta_{z}\tilde{\psi}_{k}(z)=\varepsilon_{k}^{2}qw_{k}^{2}(z).
\]
With abuse of language we set
\[
\tilde{\psi}_{k}(z)=\psi(\varepsilon_{k}w_{k}).
\]
Thus
\begin{eqnarray}
I_{\varepsilon_{k}}(u_{k}) & = & \frac{1}{2}\|u_{k}\|_{\varepsilon_{k}}^{2}-\frac{1}{p}|u_{k}^{+}|_{\varepsilon_{k},p}^{p}+\frac{\omega}{4}\frac{1}{\varepsilon_{k}^{3}}\int_{\Omega}qu_{k}^{2}\psi(u_{k})=\nonumber \\
 & = & \frac{1}{2}\|w_{k}\|_{H^{1}(\mathbb{R}^{3})}^{2}-\frac{1}{p}\|w_{k}^{+}\|_{L^{p}(\mathbb{R}^{3})}^{p}+\frac{\omega}{4}\int_{\Omega_{\varepsilon_{k}}}qw_{k}^{2}\psi(\varepsilon_{k}w_{k})=\label{eq:ik}\\
 & = & \frac{1}{2}\|w_{k}\|_{H^{1}(\mathbb{R}^{3})}^{2}-\frac{1}{p}\|w_{k}^{+}\|_{L^{p}(\mathbb{R}^{3})}^{p}+\varepsilon_{k}^{2}\frac{\omega}{4}\int_{\mathbb{R}^{3}}qw_{k}^{2}\psi(w_{k}):=E_{\varepsilon_{k}}(w_{k})\nonumber 
\end{eqnarray}
By definition of $E_{\varepsilon_{k}}:H^{1}(\mathbb{R}^{3})\rightarrow\mathbb{R},$
we get $E_{\varepsilon_{k}}(w_{k})\rightarrow m_{\infty}$. 

Given any $\varphi\in C_{0}^{\infty}(\mathbb{R}^{3})$ we set $\varphi(x)=\varphi(\varepsilon_{k}z+q_{k}):=\tilde{\varphi_{k}}(z)$.
For $k$ large enough we have that $\text{supp}\tilde{\varphi}_{k}\subset\Omega$
and, by (\ref{eq:ps}), that $E'_{\varepsilon_{k}}(w_{k})[\varphi]=I'_{\varepsilon_{k}}(u_{k})[\tilde{\varphi}_{k}]\rightarrow0.$
Moreover, by definiton of $E_{\varepsilon_{k}}$ and by Lemma \ref{lem:Tder}
we have 
\begin{eqnarray*}
E'_{\varepsilon_{k}}(w_{k})[\varphi] & = 
& \left\langle w_{k},\varphi\right\rangle _{H^{1}(\mathbb{R}^{3})}
-\int_{\mathbb{R}^{3}}|w_{k}^{+}|^{p-1}\varphi+\omega\varepsilon_{k}^{2}\int_{\mathbb{R}^{3}}qw_{k}\psi(w_{k})\varphi+\\
 & \rightarrow & \left\langle w,\varphi\right\rangle _{H^{1}(\mathbb{R}^{3})}-\int_{\mathbb{R}^{3}}|w^{+}|^{p-1}\varphi.
\end{eqnarray*}
Thus $w$ is a weak solution of 
\[
-\Delta w+w=(w^{+})^{p-1}\text{ on }\mathbb{R}^{3}.
\]
By Lemma \ref{lem:gamma} and by the choice of $q_{k}$ we have that
$w\ne0$, so $w>0$.

Arguing as in (\ref{eq:ik}), and using that $u_{k}\in{\mathcal N}_{\varepsilon_{k}}$
we have
\begin{eqnarray}
I_{\varepsilon_{k}}(u_{k}) & = & \left(\frac{1}{2}-\frac{1}{p}\right)\|u_{k}\|_{\varepsilon_{k}}^{2}
+\omega\left(\frac{1}{4}-\frac{1}{p}\right)\frac{1}{\varepsilon_{k}^{3}}\int_{\Omega}qu_{k}^{2}\psi(u_{k})\label{eq:ikH1}\\
 & = & \left(\frac{1}{2}-\frac{1}{p}\right)\|w_{k}\|_{H^{1}(\mathbb{R}^{3})}^{2}
 +\varepsilon_{k}^{2}\omega\left(\frac{1}{4}-\frac{1}{p}\right)\int_{\mathbb{R}^{3}}qw_{k}^{2}\psi(w_{k})\rightarrow m_{\infty}\nonumber 
\end{eqnarray}
and
\begin{eqnarray}
I_{\varepsilon_{k}}(u_{k}) & = & \left(\frac{1}{2}-\frac{1}{p}\right)|u_{k}^{+}|_{p,\varepsilon_{k}}^{p}
-\frac{\omega}{4}\frac{1}{\varepsilon_{k}^{3}}\int_{\Omega}qu_{k}^{2}\psi(u_{k})\label{eq:ikLp}\\
 & = & \left(\frac{1}{2}-\frac{1}{p}\right)|w_{k}^{+}|_{p}^{p}
 -\varepsilon_{k}^{2}\frac{\omega}{4}\int_{\mathbb{R}^{3}}qw_{k}^{2}\psi(w_{k})\rightarrow m_{\infty}.\nonumber 
\end{eqnarray}
So, by (\ref{eq:ikH1}) we have that $\|w\|_{H^{1}(\mathbb{R}^{3})}^{2}=\frac{2p}{p-2}m_{\infty}$
and that $\left(\frac{1}{2}-\frac{1}{p}\right)\|w_{k}\|_{H^{1}(\mathbb{R}^{3})}^{2}\rightarrow m_{\infty}$
and we conclude that $ $$w_{k}\rightarrow w$ strongly in $H^{1}(\mathbb{R}^{3})$.

Given $T>0$, by the definiton of $w_{k}$ we get, for $k$ large
enough
\begin{eqnarray}
|w_{k}^{+}|_{L^{p}(B(0,T))}^{p} & = 
& \frac{1}{\varepsilon_{k}^{3}}\int_{B(q_{k},\varepsilon_{k}T)}|u_{k}^{+}|^{p}dx
\le\frac{1}{\varepsilon_{k}^{3}}\int_{B(q_{k},r/2)}|u_{k}^{+}|^{p}dx\nonumber \\
 & \le & \left(1-\eta\right)\frac{2p}{p-2}m_{\infty}.\label{eq:contr}
\end{eqnarray}
Then we have the contradiction. In fact, by (\ref{eq:ikLp}) we have
$\left(\frac{1}{2}-\frac{1}{p}\right)|w_{k}^{+}|_{p}^{p}\rightarrow m_{\infty}$
and this contradicts (\ref{eq:contr}). At this point we have proved
the claim for $u\in{\mathcal N}_{\varepsilon}\cap I_{\varepsilon}^{m_{\varepsilon}+2\delta}$.
Now, by the thesis for $u\in{\mathcal N}_{\varepsilon}\cap I_{\varepsilon}^{m_{\varepsilon}+2\delta}$
and by (\ref{eq:ikLp}) we have
\[
I_{\varepsilon_{k}}(u_{k})=\left(\frac{1}{2}-\frac{1}{p}\right)|u_{k}^{+}|_{p,\varepsilon_{k}}^{p}
+O(\varepsilon^{2})\ge(1-\eta)m_{\infty}+O(\varepsilon^{2})
\]
and, passing to the limit, 
\[
\liminf_{k\rightarrow\infty}m_{\varepsilon_{k}}\ge m_{\infty}.
\]
This, combined by (\ref{eq:limsup}) gives us that 
\begin{equation}
\lim_{\varepsilon\rightarrow0}m_{\varepsilon}=m_{\infty}.\label{eq:mepsminfty}
\end{equation}
Hence, when $\varepsilon,\delta$ are small enough, 
${\mathcal N}_{\varepsilon}\cap I_{\varepsilon}^{m_{\infty}+\delta}
\subset{\mathcal N}_{\varepsilon}\cap I_{\varepsilon}^{m_{\varepsilon}+2\delta}$
and the general claim follows.
\end{proof}
\begin{prop}
There exists $\delta_{0}\in(0,m_{\infty})$ such that for any $\delta\in(0,\delta_{0})$
and any $\varepsilon\in(0,\varepsilon(\delta_{0})$ (see Proposition
\ref{prop:phieps}), for every function $u\in{\mathcal N}_{\varepsilon}\cap I_{\varepsilon}^{m_{\infty}+\delta}$
it holds $\beta(u)\in\Omega^{+}$. Moreover the composition 
\[
\beta\circ\Phi_{\varepsilon}:\Omega^{-}\rightarrow\Omega^{+}
\]
 is s homotopic to the immersion $i:\Omega^{-}\rightarrow\Omega^{+}$ \end{prop}
\begin{proof}
By Proposition \ref{prop:conc}, for any function $u\in{\mathcal N}_{\varepsilon}\cap I_{\varepsilon}^{m_{\infty}+\delta}$,
for any $\eta\in(0,1)$ and for $\varepsilon,\delta$ small enough,
we can find a point $q=q(u)\in\Omega$ such that 
\[
\frac{1}{\varepsilon^{3}}\int_{B(q,r/2)}(u^{+})^{p}>\left(1-\eta\right)\frac{2p}{p-2}m_{\infty}.
\]
Moreover, since $u\in{\mathcal N}_{\varepsilon}\cap I_{\varepsilon}^{m_{\infty}+\delta}$
we have 
\[
I_{\varepsilon}(u)=\left(\frac{p-2}{2p}\right)|u^{+}|_{p,\varepsilon}^{p}
-\frac{\omega}{4}\frac{1}{\varepsilon^{3}}\int_{\Omega}qu^{2}\psi(u)\le m_{\infty}+\delta.
\]
Now, arguing as in Lemma \ref{lem:stimaGeps} we have that 
\[
\|\psi(u)\|_{H^{1}(\Omega)}^{2}=q\int_{\Omega}\psi(u)u^{2}\le C\|\psi(u)\|_{H^{1}(\Omega)}\left(\int_{\Omega}u^{12/5}\right)^{5/6},
\]
so $\|\psi(u)\|_{H^{1}(\Omega)}\le\left(\int_{\Omega}u^{12/5}\right)^{5/6}$,
then 
\begin{eqnarray*}
\frac{1}{\varepsilon^{3}}\int\psi(u)u^{2} & \le 
& \frac{1}{\varepsilon^{3}}\|\psi\|_{H^{1}(\Omega)}\left(\int_{\Omega}u^{12/5}\right)^{5/6}
\le C\frac{1}{\varepsilon^{3}}\left(\int_{\Omega}u^{12/5}\right)^{5/3}\\
 & \le & C\varepsilon^{2}|u|_{12/5,\varepsilon}^{4}\le C\varepsilon^{2}\|u\|_{\varepsilon}^{4}\le C\varepsilon^{2}
\end{eqnarray*}
because $\|u\|_{\varepsilon}$ is bounded since $u\in{\mathcal N}_{\varepsilon}\cap I_{\varepsilon}^{m_{\infty}+\delta}$.

Hence, provided we choose $\varepsilon(\delta_{0})$ small enough,
we have 
\[
\left(\frac{p-2}{2p}\right)|u^{+}|_{p,\varepsilon}^{p}\le m_{\infty}+2\delta_{0}.
\]
 So, 
\[
\frac{\frac{1}{\varepsilon^{3}}\int_{B(q,r/2)}(u^{+})^{p}}{|u^{+}|_{p,\varepsilon}^{p}}>\frac{1-\eta}{1+2\delta_{0}/m_{\infty}}
\]
Finally, 
\begin{eqnarray*}
|\beta(u)-q| & \le & \frac{\left|\frac{1}{\varepsilon^{3}}\int_{\Omega}(x-q)(u^{+})^{p}\right|}{|u^{+}|_{p,\varepsilon}^{p}}\\
 & \le & \frac{\left|\frac{1}{\varepsilon^{3}}\int_{B(q,r/2)}(x-q)(u^{+})^{p}\right|}{|u^{+}|_{p,\varepsilon}^{p}}+\frac{\left|\frac{1}{\varepsilon^{3}}\int_{\Omega\smallsetminus B(q,r/2)}(x-q)(u^{+})^{p}\right|}{|u^{+}|_{p,\varepsilon}^{p}}\\
 & \le & \frac{r}{2}+2\text{diam}(\Omega)\left(1-\frac{1-\eta}{1+2\delta_{0}/m_{\infty}}\right),
\end{eqnarray*}
so, choosing $\eta$, $\delta_{0}$ and $\varepsilon(\delta_{0})$
small enough we proved the first claim. The second claim is standard.\end{proof}


\begin{thebibliography}{10}

\bibitem{AR}A. Ambrosetti, D. Ruiz, \emph{Multiple bound states for
the Schroedinger-Poisson problem}, Commun. Contemp. Math. \textbf{10}
(2008) 391\textendash{}404

\bibitem{ADP}A. Azzollini, P. D\textquoteright{}Avenia, A. Pomponio,
\emph{On the Schroedinger-Maxwell equations under the effect of a
general nonlinear term}, Ann. Inst. H. Poincar\'e Anal. Non Linaire \textbf{27}
(2010), no. 2, 779\textendash{}791 

\bibitem{AP}A. Azzollini, A. Pomponio, \emph{Ground state solutions
for the nonlinear Schroedinger-Maxwell equations}, J. Math. Anal.
Appl. \textbf{345} (2008) no. 1, 90\textendash{}108

\bibitem{BJL} J.Bellazzini, L.Jeanjean, T.Luo, \emph{Existence and
instability of standing waves with prescribed norm for a class of
Schroedinger-Poisson equations} in press on Proc. London Math. Soc.
(arXiv http://arxiv.org/abs/1111.4668)

\bibitem{B}V. Benci, \emph{Introduction to Morse theory: A new approach},
in: Topological Nonlinear Analysis, in: Progr. Nonlinear Differential
Equations Appl., vol. 15, Birkhauser Boston, Boston, MA, 1995, pp.
37\textendash{}177.

\bibitem{BBM}V. Benci, C. Bonanno, A.M. Micheletti, 
\emph{On the multiplicity of solutions of a nonlinear elliptic problem on Riemannian manifolds},
Journal of Functional Analysis \textbf{252} (2007) 464\textendash{}489
37\textendash{}177.

\bibitem{BC1}V. Benci, G. Cerami, \emph{The effect of the domain
topology on the number of positive solutions of nonlinear elliptic
problems}, Arch. Ration. Mech. Anal. \textbf{114} (1991) 79\textendash{}93.

\bibitem{BC2}V. Benci, G. Cerami, \emph{Multiple positive solutions
of some elliptic problems via the Morse theory and the domain topology},
Calc. Var. Partial Differential Equations \textbf{2} (1994) 29\textendash{}48.

\bibitem{BF}V.Benci, D.Fortunato, \emph{An eigenvalue problem for
the Schroedinger-Maxwell equations}, Topol. Methods Nonlinear Anal.
\textbf{11} (1998), no. 2, 283\textendash{}293

\bibitem{DM}T. D\textquoteright{}Aprile, D. Mugnai, \emph{Solitary
waves for nonlinear Klein-Gordon-Maxwell and Schroedinger-Maxwell
equations}, Proc. Roy. Soc. Edinburgh Sect. A \textbf{134} (2004),
no. 5, 893\textendash{}906.

\bibitem{DW1}T. D\textquoteright{}Aprile, J. Wei, \emph{Layered solutions
for a semilinear elliptic system in a ball}, J. Differential Equations
\textbf{226} (2006) , no. 1, 269\textendash{}294. 

\bibitem{DW2}T. D\textquoteright{}Aprile, J. Wei, \emph{Clustered
solutions around harmonic centers to a coupled elliptic system}, Ann.
Inst. H. Poincar\'e Anal. Non Lin\'eaire \textbf{24} (2007), no. 4, 605\textendash{}628.

\bibitem{GM1}M. Ghimenti, A.M. Micheletti, \emph{Number and profile
of low energy solutions for singularly perturbed Klein Gordon Maxwell
systems on a Riemannian manifold}, work in preparation.

\bibitem{IV}I. Ianni, G. Vaira, \emph{On concentration of positive
bound states for the Schroedinger-Poisson problem with potentials},
Adv. Nonlinear Stud.\textbf{ 8} (2008), no. 3, 573\textendash{}595. 

\bibitem{K}Kikuchi, \emph{On the existence of solutions for a elliptic
system related to the Maxwell-Schroedinger equations}, Non- linear
Anal. \textbf{67} (2007) 1445\textendash{}1456.

\bibitem{P}R. S. Palais, \emph{Homotopy theory of infinite dimensional
manifolds}, Topology \textbf{5} (1966), 1\textendash{}16.

\bibitem{PS}L.Pisani, G.Siciliano, \emph{Note on a Schroedinger-Poisson
system in a bounded domain}, Appl. Math. Lett. \textbf{21} (2008),
no. 5, 521\textendash{}528. 

\bibitem{R}D. Ruiz, \emph{Semiclassical states for coupled Schroedinger-Maxwell
equations}: Concentration around a sphere, Math. Models Methods Appl.
Sci. \textbf{15} (2005), no. 1, 141\textendash{}164. 

\bibitem{S}G. Siciliano, \emph{Multiple positive solutions for a
Schroedinger-Poisson-Slater system}, J. Math. Anal. Appl. \textbf{365}
(2010), no. 1, 288\textendash{}299. 

\bibitem{WZ}Z. Wang, H.S. Zhou, \emph{Positive solution for a nonlinear
stationary Schroedinger-Poisson system in $\mathbb{R}^{3}$,} Discrete
Contin. Dyn. Syst. \textbf{18} (2007) 809\textendash{}816.\end{thebibliography}
\end{document}